\documentclass[11pt]{article}
\usepackage{enumerate}
\usepackage{enumitem}
\usepackage{fixltx2e}
\usepackage{graphicx}
\usepackage{longtable}
\usepackage{float}
\usepackage{setspace}
\usepackage{wrapfig}
\usepackage{soul}
\usepackage{amsmath}
\usepackage{mathtools}
\usepackage{mathrsfs}
\usepackage{textcomp}
\usepackage{marvosym}
\usepackage{wasysym}
\usepackage{latexsym}
\usepackage{amssymb}
\usepackage{hyperref}
\usepackage{graphicx}
\usepackage{subfig}
\usepackage{longtable}
\usepackage{float}
\usepackage{array}
\usepackage[ruled]{algorithm}
\usepackage[noend]{algpseudocode}
\usepackage[table]{xcolor}
\usepackage{arydshln}
\usepackage{cancel}
\usepackage{url}
\usepackage{comment}
\usepackage[square, comma]{natbib}
\setcitestyle{numbers}
\usepackage[margin=1in]{geometry}
\usepackage[affil-it]{authblk}
\usepackage{titlesec}
\usepackage{color}
\titleformat{\paragraph}
{\normalfont\normalsize\bfseries}{\theparagraph}{1em}{}
\titlespacing*{\paragraph}
{0pt}{3.25ex plus 1ex minus .2ex}{1.5ex plus .2ex}

\input{./Definitions}

\newcommand{\dnc}{D\&C }

\renewcommand{\mb}{\mathbf{m}}

\allowdisplaybreaks

\newcommand{\var}{\text{var}}

\usepackage{enumerate}
\usepackage{enumitem}
\usepackage{fixltx2e}
\usepackage{graphicx}
\usepackage{longtable}
\usepackage{float}
\usepackage{wrapfig}
\usepackage{soul}
\usepackage{multicol}

\usepackage{mathtools}
\usepackage{mathrsfs}
\usepackage{textcomp}
\usepackage{marvosym}
\usepackage{wasysym}
\usepackage{latexsym}
\usepackage{amssymb}

\usepackage{subfig}
\usepackage{longtable}
\usepackage{float}
\usepackage{array}
\usepackage{titlesec}
\usepackage{color}
\RequirePackage{latexsym}
\RequirePackage{amsmath}
\RequirePackage{amssymb}
\RequirePackage{bm}             %
\RequirePackage{soul}
\usepackage{amsmath}
\usepackage{graphicx,psfrag,epsf}
\usepackage{enumerate}

\usepackage{url} % not crucial - just used below for the URL
\usepackage{lineno,hyperref}

\def\T{{\top}}

\title{An Algorithm for Distributed Bayesian Inference in Generalized Linear Models}

\author[1]{N. D. Shyamalkumar \thanks{\url{shyamal-kumar@uiowa.edu}}}
\author[1]{Sanvesh Srivastava \thanks{\url{sanvesh-srivastava@uiowa.edu}}}
\affil[1]{Department of Statistics and Actuarial Science, The University of Iowa}

\date{\today}

\begin{document}

\maketitle

\begin{abstract}
  Monte Carlo algorithms, such as Markov chain Monte Carlo (MCMC) and Hamiltonian Monte Carlo (HMC), are routinely used for Bayesian inference in generalized linear models; however, these algorithms are prohibitively slow in massive data settings because they require multiple passes through the full data in every iteration. Addressing this problem, we develop a scalable extension of these algorithms using the divide-and-conquer (D\&C) technique that divides the data into a sufficiently large number of subsets, draws parameters in parallel on the subsets using a \textit{powered} likelihood, and produces Monte Carlo draws of the parameter by combining parameter draws obtained from each subset. These combined parameter draws play the role of draws from the original sampling algorithm. Our main contributions are two-fold. First, we demonstrate through diverse simulated and real data analyses that our distributed algorithm is comparable to the current state-of-the-art \dnc algorithm
  in terms of statistical accuracy and computational efficiency. Second, providing theoretical support for our empirical observations, we identify regularity assumptions under which the proposed algorithm leads to asymptotically optimal inference. We illustrate our methodology through normal linear and logistic regressions, where parts of our \dnc algorithm are analytically tractable. 
\end{abstract}
  
\noindent%

\maketitle

\section{Introduction}

% - what is the problem?
Generalized linear models (GLMs) are widely used for regression and classification tasks. There are a variety of Bayesian and frequentist approaches for fitting GLMs. Our focus is on posterior inference in Bayesian GLMs using Monte Carlo algorithms, such as MCMC and HMC \citep{Geletal14b}. These algorithm bypass asymptotic approximations and can be easily implemented using standard software; for example, Stan and JAGS \citep{Caretal17,Plu03}. Unfortunately, Monte Carlo algorithms are inefficient in applications involving massive data because they generate a latent random variable specific to every sample in every iteration. This has motivated a rich literature on distributed Bayesian inference for scaling existing sampling algorithms to massive data settings using the \dnc technique. We develop a distributed Bayesian approach for posterior inference in GLMs based on an approximation to the Wasserstein barycenter \citep{AguCar11,Srietal15}. Our algorithm scales well in massive data settings 
because it is employed in parallel on many smaller subsets and is easily implemented in practice. 

% - why is it important?
GLMs are extremely popular in a variety of fields \citep{GelHil07}. For example, logistic regression is widely used for regression problems with a binary or multiple-class responses, and negative binomial regression is extremely popular in applications where the responses are in the form of discrete counts \citep{Far16}. In modern applications, however, there has been an explosion in the number of observations; for example, it is common to collect millions of categorical responses daily in the form of ``likes'' on  Facebook, ``retweets'' on Twitter, and ratings on online movie databases and vendors. These applications motivate the development of automated methods for fitting GLMs that are easily implemented and leverage parallel computing. There are many optimization-based methods to accomplish the desired goal, but very few such methods exist under the Bayesian paradigm. 

% - what has been done before?
The inefficiency of posterior computations in massive data settings has motivated significant interest in developing general and scalable Bayesian sampling algorithms. The literature is still developing, but three main techniques are at the forefront. The first relies on analytic approximations of the posterior, such as expectation propagation (EP), variational Bayes (VB), and Laplace approximation \citep{Rueetal09, Geletal14, TanNot14, Kucetal15,LeeWan16,Ragetal16}. These approximations underestimate posterior uncertainty unless proper care is exercised \citep{Giaetal18}, whereas MCMC and HMC are known to be accurate under general assumptions. The second technique uses subsampling or efficient approximations of transition kernels to avoid computational bottlenecks in MCMC or sequential Monte Carlo (SMC) algorithms \citep{WelTeh11,AhnKorWel12,KorCheWel13,Lanetal14,Shaetal14,MacAda15,Baretal15,Johetal15,Alqetal16,CamBro18,Quietal18,Quietal19}. These algorithms provide reliable posterior uncertainty estimates if their tuning parameters are chosen properly; therefore, Gibbs sampling and the HMC algorithm as implemented in Stan are more suited for automated applications because they are free of any proposal tuning.  

%- why hasn't been solved yet? or what are the problems with existing approaches?
The third group of methods is based on the \dnc technique. This technique is not new, but its application for scalable Bayesian inference is recent. The methods in this group operate in three steps: randomly divide the full data into smaller subsets, run a modified form of an existing sampling algorithm in parallel on all the subsets, and combine the parameter draws from all the subsets. In the second step, the prior \citep{Scoetal16} or the likelihood \citep{Minetal14} is modified. Current methods ensure that the full data posterior and the probability distribution estimated in the combination step lead to the same inference in terms of parameter and uncertainty estimates \citep{NeiWanXin13,WanDun13,Wanetal15,Minetal17}. The Wasserstein Posterior (WASP) is one such method that modifies the likelihood and combines the posterior distributions estimated on the subsets through their barycenter \citep{Srietal15}. A criticism of the WASP is that its combination algorithm requires solving a computationally expensive linear program. When the interest lies in a one-dimensional functional of the parameters, then the Posterior Interval Estimation (PIE) algorithm circumvents this issue by exploiting the analytic form of the barycenter in terms of quantiles of the posterior distributions estimated on the subsets \citep{Lietal16,SavSri16,Guhetal17}. Another alternative is the Double Parallel Monte Carlo (DPMC) algorithm that approximates the full data posterior using a mixture distribution obtained by a common centering of the subset MCMC draws \citep{XueLia17}. 

The computation of the Wasserstein barycenter is greatly simplified if the posterior distribution belongs to a location-scatter family of distributions \citep{Alvetal16}. \citet{XuSri20} have developed this idea further for scalable inference in linear mixed-effects models using WASP. Their algorithm transforms parameter draws from all the  subsets into Monte Carlo draws from the WASP by a simple centering and scaling operation. This results is massive gains in computational efficiency while retaining the simplicity of DPMC algorithm. In this work, we show that this \dnc algorithm extends to a broader class of models, including GLMs. We identify regularity assumptions on the likelihood, subset size, and the number of subsets for which the proposed \dnc algorithm is theoretically valid and demonstrate numerically that the proposed \dnc algorithm has similar performance as the current state-of-the-art DPMC algorithm on simulated and a movie ratings data in terms of speed and accuracy. 

\section{Distributed Bayesian inference in GLMs}

\subsection{First step: Creation of data subsets}
\label{first-step-losp}

Consider a GLM based on the full data. Let $n$ be the sample size, $y_i$ be the $i$th response, $x_i = (x_{i1}, \ldots, x_{ip}) \in \RR^p$  be the vector of predictors for sample $i$ ($i=1, \ldots, n$), $y = (y_1, \ldots, y_n)^\T$ be the response, $X \in \RR^{n \times p}$ be the design matrix with $x^T_i$ as its $i$th row, and $\Dcal = \{y, X\}$ be the full data. Additionally, denote the regression coefficients in a GLM as $\beta = (\beta_1, \ldots, \beta_p) \in \RR^p$, $\theta$ is the vector of all parameters including $\beta$, and $u_i$ be the variables in addition to $x_i$ required to describe the distribution of $y_i$; for example, if $s_i$ is the number of trails for the $i$th sample  in Binomial regression, then $u_i = \{s_i\}$ and $\theta = \{\beta\}$. It is assumed that the samples $\{y_i, x^\T_i\}_{i=1}^n$ are independent and the $y_i$s follow the same distribution, which is a member of the exponential family. Let $\mu_i = \EE(y_i)$ be the mean of $y_i$ and $\eta_i = x_i^\T \beta$ be the $i$th linear predictor ($i=1, \ldots, n$). A \emph{link} function $g$ is chosen depending on the distribution of the responses, and a GLM posits
\begin{align}
  \label{eq:glm}
  \eta_i = g(\mu_i) = x_i^\T \beta, \quad y_i \overset{ind.}{\sim} F_{i\theta}, \quad i = 1, \ldots, n,
\end{align}
where $F_{i \theta}$ is the distribution function of $y_i$ parameterized in terms of $\mu_i = g^{-1}(x_i^\T \beta)$, $\theta$, and $u_i$. The last equation in \eqref{eq:glm} gives the likelihood of $\theta$, and it is combined with the prior density of $\theta$, $\pi(\theta)$, using Bayes rule to obtain the posterior density of $\theta$, $\pi(\theta \mid \Dcal)$.

The posterior density $\pi(\theta \mid \Dcal)$ is analytically intractable in most applications and Monte Carlo algorithms are used to drawing $\theta$ from  $\pi(\theta \mid \Dcal)$. This task is accomplished using a variety of sampling algorithms and is easily automated using robust software packages such as Stan or JAGS; see \citet{Geletal14b} for details. Unfortunately, most of these algorithm pass through the full data and  generate a latent variable specific to every sample depending on the $F_{i \theta}$ in \eqref{eq:glm}. This is time consuming in massive data applications. The \dnc approach solves this problem by running the sampling algorithm in parallel on smaller subsets created from the full data. 

The first step of the \dnc algorithm creates $k$ subsets from the full data by random subsampling. Let $\Dcal_{(j)}$ be the data on subset $j$, $m_j$ be the sample size of subset $j$, $\beta_{(j)}$, $\theta_{(j)}$, $y_{(j)}$, $X_{(j)}$, $y_{(j)i}$, $x_{(j)i}$, $u_{(j)i}$ be the equivalents of of their full data counterparts $\beta$, $\theta$, $y$, $X$, $y_i$, $x_i$, $u_i$, and $\Dcal = \cup_{j=1}^k D_{(j)}$. Similarly, the GLM in \eqref{eq:glm} on subset $j$ is modified to
\begin{align}
  \label{eq:glm-subj}
  \eta_{(j)i} = g(\mu_{(j)i}) = x_{(j)i}^\T \beta_{(j)}, \quad y_{(j)i} \overset{ind.}{\sim} F_{i \theta_{(j)}}, \quad i = 1, \ldots, m_j; \; j = 1, \ldots, k.  
\end{align}
The Monte Carlo algorithm for drawing $\theta_{(j)}$ cycles through $m_j$ samples in every iteration. If $k$ is chosen large enough such that $m_j \ll n$ for every $j$, then drawing $\theta$ in parallel on all the $k$ subsets is faster by a factor of $O(k)$ compared to the full data. On the other hand, a subset conditions on approximately $(1/k)$-fraction of the full data, so the $\theta_{(j)}$ draws overestimate posterior uncertainty compared to the $\theta$ draws from  $\pi(\theta \mid \Dcal)$. Addressing this mismatch between the posterior uncertainty estimates, the next section modifies the sampling algorithm on the subsets without compromising on its efficiency. 

\subsection{Second step: Modified sampling using a powered likelihood}
\label{second-step-losp}

The second step of our \dnc algorithm modifies the likelihood of $\theta$ before the application of sampling algorithm on any subset. Since $\Dcal_{(j)}$ has $(m_j/n)$-fraction of the samples in $\Dcal$, the asymptotic variance of the posterior distribution of $\theta$, which conditions on $\Dcal_{(j)}$, has an inflated variance by a factor of $n/m_j$ relative to that of $\pi(\theta \mid \Dcal)$. This problem is solved by raising the likelihood of $\theta$ on subset $j$ to the power of $n/m_j$, a strategy known as stochastic approximation \citep{Minetal14}. Let $\ell(\theta_{(j)})$ be the likelihood of $\theta_{(j)}$. The sampling algorithm treats $\{\ell(\theta_{(j)})\}^{n / m_j}$ as the pseudo likelihood on subset $j$ and defines the density of $\theta_{(j)}$ using Bayes rule as 
\begin{align}
  \label{eq:subsamp}
  \pi(\theta_{(j)} \mid \Dcal_{(j)}) = \frac{\{\ell(\theta_{(j)})\}^{n / m_j} \pi(\theta_{(j)})} {\int_{\RR^p} \ell(\theta_{(j)})\}^{n / m_j} \pi(\theta_{(j)}) d\theta_{(j)}}, \quad j = 1, \ldots, k,
\end{align}
where the prior is chosen such that $\int \ell(\theta_{(j)})\}^{n / m_j} p(\theta_{(j)}) d \theta_{(j)}$ is finite. Any sampling algorithm can be used to draw $\theta_{(j)}$ from the density in \eqref{eq:subsamp}, but we have used Stan to obtain posterior draws of $\theta_{(j)}$. We also provide two illustrative examples later in Section \ref{sec:illust-eg} where $\pi(\theta_{(j)} \mid \Dcal_{(j)})$ is analytically tractable.  

The posterior draws of $\theta_{(1)}, \ldots, \theta_{(k)}$ are obtained using \eqref{eq:subsamp} in parallel on the $k$ subsets. Let $T$ be the total number of post burn-in iterations on every subset, $\theta_j^{(t)}$ be the posterior draw of $\theta_{(j)}$ at the $t$th iteration, and $\Pi(\cdot \mid \Dcal_{(j)})$ be the posterior distribution of $\theta$ given $\Dcal_{(j)}$ with density $\pi(\theta_{(j)} \mid \Dcal_{(j)})$ in \eqref{eq:subsamp}. The distribution $\Pi(\cdot \mid \Dcal_{(j)})$ is called the $j$th \textit{subset posterior distribution}, $\theta_{(j)}^{(1)}, \ldots, \theta_{(j)}^{(T)}$ are called the $j$th \emph{subset posterior draws} of $\theta$, and $\theta_{(j)}^{(t)}$ is approximately distributed as $\Pi(\cdot \mid \Dcal_{(j)})$ for every $j$ and $t$ because the subset posterior draws are collected after convergence of sampling algorithm to the target distribution. The first two steps of subset samplers in the proposed \dnc algorithm {do overcome the hurdles in using} the original sampling algorithm in massive data settings while retaining its advantages. First, if we assume that each subset has sample size $m$, then subset posterior computations scale as $O(m)$, which is smaller than the complexity of original sampler by a factor of $k$. Second, the pseudo likelihood in \eqref{eq:subsamp} is a slight modification of the original likelihood, so the subset samplers are obtained using a simple modification of the original sampler.

Any existing method can be used to combine the subset posterior draws, but we focus on developing an approximation to the WASP because the linear program for its estimation has a computational complexity of $O(T^{5})$, which is prohibitively slow in practice. Our combination algorithm is based on the one in \citep{XuSri20} and is an approximation to the WASP. % We show that the empirical measures supported on the draws from the true WASP and that resulting from our algorithm are asymptotically equivalent. In our simulated and real data analyses, this algorithm outperforms its competitors in accuracy and computational efficiency. On the other hand, the current DPMC algorithm is easy to implement but fails to account for the settings where the covariance matrices of the subset posterior distributions significantly differ, which often happens when the fraction of ones and zeros are imbalanced. 

\subsection{Third Step: Combination of subset draws}

\begin{enumerate}[]
  \begin{algorithm}[t]
    \caption{The \dnc algorithm based on a WASP approximation.} \label{losp-sampler}
  \item {\bf INPUT} 
    \begin{enumerate}
    \item Subset posterior draws for $\theta$, $\theta_{(j)}^{(t)}$ ($j=1, \ldots, k$; $t = 1, \ldots, T$), and a known function of $\theta$, $f(\theta)$. 
    \item Mean vectors and covariance matrices of the subset posterior distributions and the approximate WASP posterior distribution, $\hat \mu_{(j)}$, $\hat {\overline \mu}$, $\hat \Sigma_{(j)}$, $\hat {\overline \Sigma}$ ($j=1, \ldots, k$). 
    \end{enumerate}
  \item {\bf DO}
    \begin{enumerate}
    \item Center and scale the subset posterior draws for $j = 1, \ldots, k$ and $t = 1, \ldots, T$ to define $\hat q_{(j)}^{(t)} = \hat \Sigma_{(j)}^{-1/2} (\theta_{(j)}^{(t)} - \hat \mu_{(j)})$.
    \item For $j = 1, \ldots, k$ and $t=1, \ldots, T$, define the $t'$th combined draw using the $t$th subset posterior draw of $\theta$ as       
      \begin{align}
        \label{eq:alg3}
        \overline \theta^{(t')} = \hat {\overline \mu} + \hat {\overline \Sigma}^{1/2}  \hat q_{(j)}^{(t)}, \quad t' = (j-1)T + t, 
      \end{align}
      where $\overline \theta^{(t')}$ is the approximate WASP draw for $\theta$.
    \item For $t' = 1, \ldots, kT$, define the $t'$th approximate WASP draw for $f(\theta)$ as  $f(\overline \theta^{(t')})$. 
    \end{enumerate}
  \item {\bf RETURN} \\ \vspace{5pt}
    $\overline \theta^{(1)}, \ldots, \overline \theta^{(kT)}$, $f(\overline \theta^{(1)}), \ldots, f(\overline \theta^{(kT)})$ as the approximate WASP draws. 
  \end{algorithm}
\end{enumerate}

We require two concepts from the theory of optimal transport for approximating the WASP. The first is that of the Wasserstein barycenter. Let $\| \cdot \|_2$ be the Euclidean metric, $\Pcal(\RR^p)$ be the set of all probability measures on $\RR^p$, $\Pcal_2(\RR^p)$ denote the Wasserstein space of order 2 given by $\{\nu \in \Pcal(\RR^p): \int \| \theta \|_2^2 \nu(d \theta) < \infty \}$, and the Wasserstein distance of order 2 between $\nu_1, \nu_2 \in \Pcal_2(\RR^p)$ given by $\underset{\pi \in \Lcal (\nu_1, \nu_2)} {\mathrm{inf}} \, \left(\int_{\RR^p \times \RR^p} \| x - y \|^2_2 \, d \pi(x, y)\right)^{1/2}$, where $\Lcal (\nu_1, \nu_2)$ is the set of all probability measures on $\RR^p \times \RR^p$ with marginals $\nu_1$ and $\nu_2$, be denoted by $W_2(\nu_1, \nu_2)$. Assume that $\nu_1, \ldots, \nu_k \in \Pcal_2(\RR^p)$, then their Wasserstein barycenter with weights $(w_1, \ldots, w_k)$ equals
\begin{align}
  \label{eq:wa1}
  \overline \nu = \underset{\nu \in \Pcal_2(\RR^p)}{\argmin} \, \sum_{j=1}^k {w_j} W_2^2(\nu,  \nu_j) , \quad \sum_{j=1}^k w_j = 1, \quad w_1, \ldots, w_k > 0,   
\end{align}
where $\overline \nu$ exists uniquely \citep{AguCar11}. 
In scalable Bayesian applications, $\Pi(\cdot \mid \Dcal_{(1)}), \ldots, \Pi(\cdot \mid \Dcal_{(k)})$ play the role of $\nu_1, \ldots, \nu_k$, respectively.  Their Wasserstein barycenter is the WASP, denoted as $\overline \Pi (\cdot \mid \Dcal)$, and replaces $\Pi(\cdot \mid \Dcal)$ for inference on $\theta$ \citep{Srietal15}. The optimization problem in \eqref{eq:wa1} is posed as a linear program in terms of empirical measures supported on subset posterior draws 
and efficient algorithms exist to obtain an empirical approximation of $\overline \Pi (\cdot \mid \Dcal)$ \citep{Lietal16,Staetal17}. We fix $w_j $ at $ 1/k$ and assume that $\Pi(\cdot \mid \Dcal_j) \in \Pcal_2(\RR^p)$ ($j=1, \ldots, k$), which implies that $\overline \Pi (\cdot \mid \Dcal) \in \Pcal_2(\RR^p)$.

The second concept is the location-scatter family of probability measures. It is defined as follows:
\begin{definition}[Location-scatter family; \citet{Alvetal16}]\label{loc-sca}
  Let $X_0$ be a random vector with probability law $G_0 \in \Pcal_2(\RR^p)$ such that $E(X_0) = 0$ and $\var(X_0) = I$, where $I$ is a $p \times p$ identity matrix, $\Lcal(W)$  be the probability distribution of a random variable $W$, and  $\Mcal^{p \times p}_{+}$ be the set of $p \times p $ positive definite matrices. The family $\Fcal(G_0) = \{\Lcal(\Sigma^{1/2} X_0 + \mu):  \Sigma \in \Mcal^{p \times p}_{+}, \mu \in \RR^p\}$ of probability laws induced by positive definite affine transformations from $G_0$ is called a location-scatter family, where $\Sigma^{1/2}$ is the symmetric square-root of $\Sigma$. 
\end{definition}
The family $\Fcal(G_0)$ contains distributions parameterized in terms of their mean $\mu$ and covariance matrix $\Sigma$; elliptical families are canonical examples. Theorem 4.2 in \citet{Alvetal16} implies that if $\nu_1, \ldots, \nu_k \in \Fcal(G_0)$ for some $G_0$ and $a_j$ and $B_j$ are the mean vector and covariance matrix of $\nu_j$ ($j=1, \ldots, k$), then their Wasserstein barycenter, denoted as $\overline \nu$, also belongs to $\Fcal(G_0)$ under general assumptions and its mean vector $\overline a = \frac{1}{k} \sum_{j=1}^k a_{j}$ and the covariance matrix $\overline B$ is the limit point of the sequence $\{S_{t}\}_{t=0}^{\infty}$  defined by 
\begin{align}
  \label{eq:wa21}
  S_{t+1} = S_t^{-1/2} \left\{ \sum_{j=1}^k (1/k) \left( S_t^{1/2} B_j S_t^{1/2} \right)^{1/2} \right\}^2 S_t^{-1/2}, \quad t=0, 1, 2, \ldots, \quad S_0 = I_p. 
\end{align}

The third step of our \dnc algorithm defines the mean vector and covariance matrix of the combined posterior distribution based on the results for the location-scatter family. Let $\mu_{(j)}$, $\overline \mu$ and $\Sigma_{(j)}$, $\overline \Sigma$ be the mean vectors and covariance matrices of $\Pi(\cdot \mid \Dcal_{(j)})$ and $\overline \Pi(\cdot \mid \Dcal)$, respectively. We define $\overline \mu = \frac{1}{k} \sum_{j=1}^k \mu_{(j)}$ and $\overline \Sigma$ as the limit of the sequence  $\{\overline \Sigma_{t}\}_{t \geq 0}$, the latter defined using a numerically stable version of \eqref{eq:wa21} (see \citet{XuSri20} for further details) as,
\begin{align}
  \label{eq:wa2}
  \overline \Sigma_{t+1} = \overline  \Sigma_t^{-1/2} \left\{ \frac{1}{k}\sum_{j=1}^k \left( \overline \Sigma_t\Sigma_{(j)} \right)^{1/2} \right\}
  \left\{ \frac{1}{k}\sum_{j=1}^k \left( \overline \Sigma_t \Sigma_{(j)} \right)^{1/2} \right\}^\T
  \overline  \Sigma_t^{-1/2}, \quad t=0, 1, 2, \ldots, \quad \overline \Sigma_0 = I_p . 
\end{align}
%Note that the above approach for computing $\overline \Sigma_{t+1}$ in \eqref{eq:wa2} 
%is based on a rearrangement of matrix operations in \eqref{eq:wa21}. \citet{XuSri20} have addressed numerical instability of \eqref{eq:wa21} by relying on geometric mean of two matrices, which permits a simple expression. 
In practice, $\mu_j$s and $\Sigma_j$s are unknown, so our \dnc algorithm replaces them by their Monte Carlo estimates based on the subset posterior draws of $\theta$ as 
\begin{align}
  \label{eq:wa3}
  \hat \mu_{(j)} = \frac{1}{T}\sum_{t=1}^T \theta_{(j)}^{(t)}, \quad \hat \Sigma_{(j)} = \frac{1}{T}\sum_{t=1}^T (\theta_{(j)}^{(t)} - \hat \mu_{(j)}) (\theta_{(j)}^{(t)} - \hat \mu_{(j)})^\T, \quad 
  \hat{\overline \mu} = \frac{1}{k} \sum_{j=1}^k \hat \mu_{(j)}, \quad
  \hat {\overline \Sigma} = \hat {\overline \Sigma}_{\infty},
\end{align}
where $\{\hat {\overline \Sigma}_{t}\}_{t \geq 0}$ is obtained by replacing $\Sigma_{(j)}$ with $\hat \Sigma_{(j)}$ in \eqref{eq:wa2}. 

While there is no guarantee that $\Pi(\cdot \mid \Dcal_{(1)}), \ldots, \Pi(\cdot \mid \Dcal_{(k)})$ are members of a location-scatter family, our \dnc algorithm assumes this to be true and defines the mean vector and covariance matrix of the combined posterior as $\hat{\overline \mu}$ and $\hat {\overline \Sigma}$. This suggests a simple algorithm for transforming the subset posterior draws into draws from the combined posterior: (i) center and scale $j$th subset posterior draws as $\hat q_{(j)}^{(t)} = \hat {\Sigma}^{-1/2}_{(j)} (\theta_{(j)}^{(t)} - \hat \mu_{(j)}) $ and (ii) rescale and recenter $q_{(j)}^{(t)}$s to obtain draws following the combined posterior distribution as $\hat {\overline \mu} + \hat{\overline \Sigma}^{1/2} q^{(t)}_{(j)}$ for $t=1, \ldots, T$ and $j=1, \ldots, k$. This heuristic is summarized in Algorithm \ref{losp-sampler} and justified theoretically in the Section \ref{sec:theor-prop-losp}; therefore, our \dnc  algorithm provides an approximation to the WASP that reduces to the true WASP if the subset posterior distributions belong to a common location-scatter family.  

\section{Illustrative examples}
\label{sec:illust-eg}

Before discussing the theoretical properties of Algorithm \ref{losp-sampler}, we provide two illustrative examples from normal linear regression and logistic regression using Polya-Gamma data augmentation (PG-DA) \citep{Poletal13}. The subset posterior densities are analytically tractable in both examples, but Algorithm \ref{losp-sampler} is exact only in the first. For a simplified presentation, we assume that $m_1 = \cdots = m_k = m$ and $n = km$ in both examples; that is, we assume that the full data have been partitioned into disjoint subsets of equal sample size. 

\subsection{Normal linear regression}
\label{sec:linear-regression}

Consider normal linear regression model with the identity link function. Setting $g(\mu) = \mu$, $\theta = \{\beta, \sigma^2\}$, and $F_{i \theta}$ to be the Gaussian distribution with mean $x_i^\T \beta$ and variance $\sigma^2$ in \eqref{eq:glm}, we obtain that
\begin{align}
  \label{eq:mot1}
  y_i = x_i^\T \beta + \epsilon,  \epsilon \sim N(0, \sigma^2), \quad \pi(\beta, \sigma^2) \propto \sigma^{-2}, 
\end{align} 
where $u_i$ includes the variance of $\epsilon_i$ denoted as $\sigma^2$ and $(\beta, \sigma^2)$ are assigned an improper prior that maintains posterior propriety. The main interest lies in the posterior distribution of $\beta$ given $\Dcal$, which is
\begin{align}
  \label{eq:mot2}
  % \sigma^2 \mid \Dcal \sim \frac{\| y - \hat y \|_2^2}{\chi^2_{n-p}}, \quad \beta \mid \sigma^2, \Dcal \sim N\left\{ \hat \beta, \sigma^2 \left( X^\T X \right)^{-1} \right\}, \quad
  \beta \mid \Dcal \sim t_{n-p} \left\{ \hat \beta, s^2 (X^\T X)^{-1}  \right\} , \quad
  \hat \beta = \left( X^\T X \right)^{-1} X^\T y, \quad s^2 = \frac{\| y - \hat y \|_2^2}{n-p}, \quad \hat y = X \hat \beta,
\end{align}
where $n - p > 2$ and $t_{\nu}(a, A)$ is the multivariate $t$ distribution with $\nu$ degrees of freedom and ($a$, $A$) are the location and correlation parameters \citep{Geletal14b}. An application of this result in \eqref{eq:subsamp} implies that the subset posterior density of $\beta_{(j)}$ after stochastic approximation is 
\begin{align}
  \label{eq:mot2j}  
  \beta_{(j)} \mid \Dcal_{(j)} \sim t_{km-p} \left\{ \hat \beta_{(j)}, s_{(j)}^2 (X^\T_{(j)} X_{(j)})^{-1}  \right\}, \quad
  \hat \beta_{(j)} = \left( X^\T_{(j)} X_{(j)} \right)^{-1} X^\T_{(j)} y_{(j)},\quad s_{(j)}^2 = \frac{k\| y_{(j)} - \hat y_{(j)} \|_2^2}{km-p}, \quad \hat y_{(j)} = X_{(j)} \hat \beta_{(j)},
\end{align}
where the posterior distribution of $\beta_{(j)}$ given $\Dcal_{(j)}$ is called the $j$th subset posterior distribution. We have assumed that $km=n$, so the degrees of freedom of the full-data and subset posterior distributions of $\beta$ are $n-p$ and they differ only in their location and correlation parameters. 

The full-data and subset posterior distributions belong to the location scatter family specified by setting $G_0$ in Definition \ref{loc-sca} to be $t_{n-p}(0, \tfrac{n-p-2}{n-p}I_p)$. Let $\mu$, $\mu_{(j)}$ and $\Sigma$, $\Sigma_{(j)}$ be the means and covariance matrices of the full data and $j$th subset posterior distributions and $X_0 \sim G_0$. Then, $\beta \mid \Dcal$ in \eqref{eq:mot2} and $\beta_{(j)} \mid \Dcal_{(j)}$ in \eqref{eq:mot2j}, respectively, are represented in terms of $X_0$ as $\beta = \mu + \Sigma^{1/2} X_0$ and $\beta_{(j)} = \mu_{(j)} + \Sigma^{1/2}_{(j)} X_0$ \citet[Section 10]{NadKot05}; therefore, the WASP of the $k$ subset posterior distribution is $t_{n-p}(\overline \mu, \overline V)$, where $\overline \mu, \overline V$ satisfy
\begin{align}
  \label{eq:mot3}
  \overline \mu = \frac{1}{k} \sum_{j=1}^k \mu_{(j)}=\frac{1}{k} \sum_{j=1}^k \hat \beta_{(j)}, \quad 
  \overline V = \frac{1}{k} \sum_{j=1}^k  \left( \overline V^{1/2} V_{(j)} \overline V^{1/2} \right)^{1/2}, \quad V_{(j)} = s_{(j)}^2 (X^\T_{(j)} X_{(j)})^{-1}, 
\end{align}
and $\overline V$ is found using the fixed point algorithm in \eqref{eq:wa2}.

The analytic expressions of the subset posterior distributions in \eqref{eq:mot2j} enable comparisons with the full data posterior distribution under certain assumptions. Assume that $\beta_0, \sigma_0^2$ are the true parameter values in \eqref{eq:mot1}, $\{(x_i, y_i)\}_{i=1}^n$  are independent and identically distributed copies of $(x, y)$, $P_0$, $P_{y|x}$, and $P_{x}$ are the true distributions of $(x, y)$, $y$ given $x$, and $x$, $E_{0}$,  $E_{y|x}$,  $E_{x}$ are the expectations with respect to $P_{0}, P_{y|x}$, and $P_{x}$, and $V_x = E_x(x_1 x_1^\T)$ is non-singular. The residual error variance is an unbiased estimator of $\sigma_0^2$, so $E_0 (s^2) = \sigma_0^2$ and  $E_0(s_{(j)}^2)  = {k(m-p)} \sigma_0^2 / (km-p)$. The law of large numbers and Slutsky's theorem imply that
\begin{align}
  \label{eq:mot31}
  \frac{1}{s^2} \frac{X^\T X}{n} = \frac{V_x}{\sigma_0^2 } + o_n(1), \quad 
  \frac{1}{s^2_{(j)}} \frac{X^\T_{(j)} X_{(j)}}{m} = \frac{1 - o_n(1)}{1 - o_m(1)} \left\{ \frac{V_x}{\sigma_0^2 }+ o_m(1) \right\} = \frac{V_x}{\sigma_0^2 } + o_m(1), 
\end{align}
where $o_n(1)$ and $o_m(1)$ tend to 0 as $n$ and $m$ tend to infinity with $P_0$-probability 1; therefore, the conditions of Theorem 1 in \citet{Srietal18} are satisfied and $E_0 \{\sqrt{n} W_2(\beta, \overline \beta)\}^2 = o_m(1)$, where $\beta$ and $\overline \beta$ are random variables following the full-data posterior and WASP distributions. This implies that the 
WASP-based credible intervals for quantifying posterior uncertainty match with those obtained from the full data posterior distribution up to $o(1)$ terms in $P_0$-probability as $n \rightarrow \infty$; see Theorem 1 in \citet{Lietal16}.

The previous theoretical analysis suggests a simple scheme for posterior inference on $\beta$ in \eqref{eq:mot1} using Algorithm \ref{losp-sampler} when $n$ is large. Divide the $n$ samples randomly into $k$ subsets of almost equal size. Compute $ \hat \beta_{(j)}$ and $V_{(j)}$ in \eqref{eq:mot3} for $j=1, \ldots, k$ in parallel. Generate $q_1, \ldots, q_T$ independently from $t_{n-p}(0, I_p)$. Define the $t$th WASP draw as $\overline \mu + \overline V^{1/2} q_t$ ($t=1, \ldots, T$) and use them for posterior inference on $\beta$ instead of draws from the full data posterior distribution, where $\overline \mu$ and $\overline V$ are defined in \eqref{eq:mot3} and $\overline V^{1/2}$ is the symmetric square root of $\overline V$. This idea has motivated the development of the Location-Scatter WASP for linear-mixed effects model \citep{XuSri20}. Our goal in this work is to develop this idea more broadly with rigorous theoretical guarantees; see Section \ref{sec:theor-prop-losp}. 

\subsection{Logistic regression via Polya-Gamma data augmentation}
\label{log-reg}

Logistic regression is also a special case of \eqref{eq:glm}. Set $g(\mu) = \log \{\mu / (1 - \mu)\}$, $\theta = \beta$, and $F_{i \theta}$ to be the binomial distribution with mean $\mu_i:=s_ip_i$, where $p_i := (1 + e^{-x_i^\T \beta})^{-1}$ in \eqref{eq:glm}. The variable $u_i$ includes the number of trials $s_i$, $y_i$ follows $\text{Binom}(s_i, p_i)$ independently for $i=1, \ldots, n$, and $\beta$ is assigned the $N(\mu_{\beta}, \Sigma_{\beta})$ prior. The posterior distribution of $\beta$ is analytically intractable, but the P\'olya-Gamma Data Augmentation (PG-DA) strategy for logistic regression permits analytically tractable full conditional distributions \citep{Poletal13}. The PG sampler cycles between 
\begin{enumerate}[label={\arabic*},ref=\arabic*]
\item\label{s1} draw $\omega_i$ given $\beta$ and $\Dcal$ from $\text{PG}$($s_i$, $|x_i^T \beta|$) for $i = 1, \ldots, n$, where PG is the Polya-Gamma distribution; and
\item\label{s2} draw $\beta$ given $\omega_1, \ldots, \omega_n$, and  $\Dcal$ from $N$($m_{\omega}$, $V_{\omega}$), where $V_{\omega} = (X^T \Omega X + \Sigma_{\beta}^{-1})^{-1}$, $m_{\omega} = V_{\omega} (X^T \kappa + \Sigma_{\beta}^{-1} \mu_{\beta})$, $\kappa = (y_1 - s_1 / 2, \ldots, y_n - s_n / 2)$, and $\Omega$ is the diagonal matrix of $\omega_i$s. 
\end{enumerate}
The Markov chain $\Phi=\{\beta^{(t)}\}_{t=1}^{\infty}$ of the draws collected in step \ref{s2}, where $t$ indexes the iterations, has $\pi( \beta \mid \Dcal)$ as its invariant density \citep{ChoHob13}. The key idea in the PG-DA strategy is that the conditional density of $\beta$ given $\omega_1, \ldots, \omega_n$ and $y$ is
\begin{align}
  \label{eq:pg1}
  p(\beta \mid \omega_1, \ldots, \omega_n, y) \propto \prod_{i=1}^n p(y_i \mid \omega_i, \beta) p(\beta)  &\propto 
  \prod_{i=1}^n \exp \left\{ \kappa_i x_i^T \beta - \omega_i (x_i^T \beta)^2 /2 \right\} p(\beta) = \exp \left\{ - (z - X \beta)^T \Omega (z - X \beta) /2 \right\} p(\beta),
\end{align}
where $p(y_i \mid \omega_i)$ is the conditional density of $y_i$ given $\omega_i$, $\omega_i$ follows PG($s_i$, $|x_i^T \beta|$), $z = (\kappa_1 / \omega_1, \ldots, \kappa_n / \omega_n)$, $\Omega$ is defined in step \ref{s2}, and $p(\beta \mid y, \omega_1, \ldots, \omega_n) $ yields a conditionally Gaussian likelihood for $\beta$ with a working response $z$, design matrix $X$, and covariance matrix $\Omega^{-1}$. 

The subset posterior density in \eqref{eq:subsamp} for logistic regression is derived by modifying step \ref{s2} of the original PG sampler. The first step for generating $\omega_{(j)1}, \ldots, \omega_{(j)m}$ on subset $j$ is identical to step \ref{s1}: 1. draw $\omega_{(j)i}$ given $\beta_{(j)}$ and $\Dcal_{(j)}$ from $\text{PG}$($s_{(j)i}$, $x_{(j)i}^T \beta_{(j)}$) for $i = 1, \ldots, m$. 
Using \eqref{eq:pg1}, we have that 
\begin{align}
  \label{eq:pg2}
  \{p(\beta_{(j)} \mid y_{(j)}, \omega_{(j)1}, \ldots, \omega_{(j)m})\}^{n / m} \equiv \ell(\beta_{(j)} \mid y_{(j)}, \omega_{(j)1}, \ldots, \omega_{(j)m}) \propto e^{- (n / m)(z_{(j)} - X_{(j)} \beta_{(j)})^T \Omega_{(j)} (z_{(j)} - X_{(j)} \beta_{(j)}) /2} ,
\end{align}
where $\ell(\beta_{(j)} \mid y_{(j)}, \omega_{(j)1}, \ldots, \omega_{(j)m})$ is the stochastically-approximated conditionally Gaussian likelihood of $\beta_{(j)}$, its integral with respect to $p(\beta_{(j)})$ is finite, $z_{(j)} = (z_{(j)1}, \ldots, z_{(j)m})$, $z_{(j)i} = \kappa_{(j)i} / \omega_{(j)i}$, $\kappa_{(j)i} = y_{(j)i} - s_{(j)i}/2$, and $\Omega_{(j)}$ is the diagonal matrix of $w_{(j)i}$s. The $N(\mu_{\beta}, \Sigma_{\beta})$ prior on $\beta_{(j)}$ and the Bayes rule with the conditional likelihood in \eqref{eq:pg2} gives the equivalent  of step \ref{s2} in our \dnc algorithm on subset $j$: 2. draw $\beta_{(j)}$ given $\omega_{(j)1}, \ldots, \omega_{(j)m}$ and $\Dcal_{(j)}$ from $N$($m_{\omega_j}$, $V_{\omega_j}$), where $V_{\omega_j} = (\tfrac{n}{m}X_{(j)}^T \Omega_{(j)} X_{(j)} + \Sigma_{\beta}^{-1})^{-1}$, $m_{\omega_j} = V_{\omega_j} (\tfrac{n}{m} X_{(j)}^T \kappa_{(j)} + \Sigma_{\beta}^{-1} \mu_{\beta})$ and $\kappa_{(j)} = (\kappa_{(j)1}, \ldots, \kappa_{(j) m})$.

The full conditional distribution of $\beta$ after stochastic approximation is Gaussian on any subset. Unlike the previous example, this does not imply that the full-data and subset  posterior distributions of $\beta$ given $\Dcal$ belong to the same location-scatter family; therefore, the WASP is analytically intractable and the computationally expensive linear program must be employed for estimating the true WASP, which approximates the full-data posterior distribution. If we employ Algorithm \ref{losp-sampler} for combining draws of $\beta$ obtained using steps 1 and 2, then we approximate the true WASP using the barycenter of the approximations of subset posterior distributions based on a location-scatter family. This approximation to the true WASP delivers excellent performance in terms of approximating the full-data posterior distribution if the sample size on every subset is large enough to justify the Bernstein-von Mises theorem. The next section justifies this heuristic theoretically for a large class of likelihoods. 

\section{Theoretical Properties}
\label{sec:theor-prop-losp}

The previous section presented illustrative examples for linear and logistic regressions, but Algorithm \ref{losp-sampler} with a suitably replaced subset sampling scheme makes it applicable to a broad class of likelihoods. We show in this section that only the geometric ergodicity of the Monte Carlo algorithm is relevant in the theoretical analysis of Algorithm \ref{losp-sampler}. For these reasons, our theoretical analysis is stated in the setting of a broad class of likelihoods, with the geometric ergodicity of the subset samplers as a requirement. Henceforth, Algorithm \ref{losp-sampler} is to be understood in the above stated broader setting with a suitable subset sampling scheme. We start by stating the assumptions required for the theoretical validity of Algorithm \ref{losp-sampler}. In the following assumptions, $\theta_0$  is the true value of $\theta$ and $P^n_{\theta}$ is the joint distribution of the training data based a likelihood:
\begin{enumerate}
\item \label{ab} The $y_1, \ldots, y_n$ are independent and identically distributed as $P_{\theta_0}$. 
\item \label{a1} The subset posterior and full data posterior distributions belong to a location scatter family with $P_{\theta_0}^n$-probability 1. 
\item \label{xl1} The regularity assumptions of Laplace approximation hold. Let $B_{\delta}(\theta)$ denote an open ball of radius $\delta$ centered at $\theta$. Let the log likelihood of $\theta$ given $y_1, \ldots, y_n$ be $\ell_n(\theta)$, $\hat \theta_n$ be the maximum likelihood estimate of $\theta$, and $D^2 \ell_n(\theta)$ be its Hessian at $\theta$. Further, suppose that there exists positive numbers $\epsilon$, $M$, and $\eta$ and an integer $n_0$ such that for all $n \geq n_0$: (a) for every $\theta \in B_{\epsilon}(\hat \theta_n)$ and all $1 \leq j_1, \ldots, j_d \leq p$ with $1 \leq d \leq 6$, $|\partial_{j_1, \ldots, j_d} \ell_n(\theta) | < M$; (b) $\det \{D^2 \ell_n(\theta)\} > \eta$; and (c) for every $\delta$ satisfying $0 < \delta < \epsilon$, $ B_{\delta}(\hat \theta_n) \subseteq \RR^p $ and
  $$\underset{n \rightarrow \infty}{\limsup} \underset{\theta \in \RR^p \backslash B_{\delta}(\hat \theta_n) }{\sup} \{ \ell_n(\hat \theta_n) - \ell_n(\theta) \} < 0$$
\item \label{xl3} The number of subsets $k$ satisfies $k = O(1)$, and the subsets are disjoint and equal in size such that $km = n$, where we have assumed that $m=m_1 = \cdots = m_k$. 
\item \label{mc1} The number of iterations $T$ satisfies $n = o(T^{1/2})$ and $\sqrt{T} (\hat \mu_{(j)} - \mu_{(j)}) = O_Q(T^{-1/2})$ and $\sqrt{T} (\hat \Sigma_{(j)} - \Sigma_{(j)}) = O_Q(T^{-1/2})$ ($j=1, \ldots, k$), where $\hat \mu_{(j)}$ and $\hat \Sigma_{(j)}$ are defined in 
\eqref{eq:mcmceq} and 
 {$Q$ is the true joint probability measure on $\theta_{(1)}, \ldots, \theta_{(j)}$ draws defined in Theorem \ref{mc-err}}.
\end{enumerate}
%Assumptions \ref{ab} and \ref{a11} are required to justify asymptotic normality of the maximum likelihood estimate of $\theta$; see Example 5.40 in \citet{Van00b}. 

Assumptions \ref{ab}--\ref{xl3} are commonly assumed in \dnc Bayesian inference and known to be satisfied if $P_{\theta_0}$ is a member of the exponential family; see Theorem 1 in \citet{XueLia17}. Assumption \ref{a1} is required for obtaining an analytic expression for the $W_2$-distance between the full data posterior distribution and the approximate WASP. Assumption \ref{xl1} is based on those required for the validity of the Laplace approximation for the full data and subset posterior distributions; see \citet{Kasetal90}. Our results generalize to cases where the subset sizes differ, but requiring a common subset sample size in Assumption \ref{xl3}  simplifies the analysis. We have also assumed that $k = O(1)$ for a simplified analysis, but this assumption can be relaxed using the theoretical setup in Theorem 1 of \citet{Lietal16}. Assumption \ref{mc1} is satisfied when the subset sampling scheme is geometrically ergodic; for example, Proposition 3.1 in \citet{ChoHob13} shows this for the PG-DA strategy. 

Let   
$\Pi$ be the full-data posterior, $\overline \Pi$ be the combined posterior in Algorithm \ref{losp-sampler} based on the WASP approximation, and $\mu$, $\overline \mu$ and $\Sigma$, $\overline \Sigma$ be the means and covariance matrices of $\Pi$, $\overline \Pi$. 
While $\Pi$ and $\overline \Pi$ are both analytically intractable, it is more efficient to obtain $\theta$ draws from $\overline \Pi$ using Algorithm \ref{losp-sampler} than from a general sampling scheme for $\Pi$. One source of error in using draws from Algorithm \ref{losp-sampler} for posterior inference on $\theta$ is statistical in nature, which arises from the use of $\overline \Pi$ instead of $\Pi$.
We quantify this error by $W_2(\Pi, \overline \Pi)$, which is independent of the subset sampling scheme. Algorithm \ref{losp-sampler} is motivated by the fact that if $\Pi$ and $\overline \Pi$ belong to the same location-scatter family, then 
\begin{align}
  \label{eq:t1e1}
  W^2_2(\Pi, \overline \Pi) = \| \mu - \overline \mu \|_2^2 + \tr \{\Sigma + \overline \Sigma - 2 (\overline \Sigma^{1/2} \Sigma\overline \Sigma^{1/2})^{1/2}\};
\end{align}
see Theorem 2.3 in \citet{Alvetal16}. For the convenience of theoretical analysis, we make this an assumption on $\Pi$ and the subset posterior distributions.  
If we show that the two terms on the right are $o(n^{-1})$ terms in $P^n_{\theta_0}$-probability, then 
$W_2(\Pi, \overline \Pi)$ is $o(n^{-1/2})$ in $P^n_{\theta_0}$-probability, implying that the statistical error decays to 0 at the parametric optimal $n^{-1/2}$ rate. The next theorem shows that this is indeed true.% under certain regularity assumptions. 
\begin{theorem}\label{stat-ord}
  If Assumptions \ref{ab}--\ref{xl3} hold, then as $n, m \rightarrow \infty$
\begin{align*} 
  W_2(\Pi, \overline \Pi) = o(n^{-1/2}) \text{ in } P^n_{\theta_0}\text{-probability}.
\end{align*}  
\end{theorem}

In typical settings $\overline{\Pi}$ is not analytically tractable and one resorts to working with the empirical distribution constructed from the MCMC draws. Algorithm \ref{losp-sampler} provides an alternate way of arriving at an empirical distribution, which serves as an approximation to  the latter. We posit that the distribution between these two empirical distributions is a relevant measure of Monte Carlo error. Theorem \ref{mc-err}  below uses the rate of convergence of the Monte Carlo error to give a guidance on the choice of $T$. Note that these two empirical measures are random quantities, and hence in Theorem \ref{mc-err} we derive the asymptotic order for a certain coupled versions of these measures. We describe this coupling below. 

Based on Definition \ref{loc-sca}, let $G_0$ be the  distribution with mean ${0}$ and covariance matrix $I_p$ that defines the location-scatter family of subset posterior distributions, and  $\theta'_i$ ($i=1, \ldots, kT$) denote a $kT$ independent draws from $G_0$. If we assume that the subset posterior distributions belongs to the location-scatter family defined by $G_0$, then define $\theta_{(j)}^{(t)} = \mu_{(j)}+\Sigma_{(j)}^{1/2}\theta'_{(j-1)T+t}$ for $t=1,\ldots,T$ as a random sample of size $T$ from the $j$th subset posterior ($j=1,\ldots,k$), where $\theta'_i$s are unobserved and  $\theta_{(j)}^{(t)}$s are the $j$th subset MCMC draws. Denote the empirical means and the covariance matrices computed using $\theta_{(j)}^{(t)}$s as $\hat\mu_{(j)}$ and $\hat\Sigma_{(j)}$, respectively, and the uniform empirical measure supported on atoms
\begin{align}
  \label{eq:mcmceq}
  \hat{\overline{\mu}} + \hat{\overline{\Sigma}}^{\frac{1}{2}} \hat\Sigma_{(j)}^{-{1}/{2}}(\hat\mu_{(j)}-\hat{\mu}_{(j)})+\hat{\overline{\Sigma}}^{\frac{1}{2}} \hat\Sigma_{(j)}^{-1/2}\hat\Sigma_{(j)}^{\frac{1}{2}}\theta'_{(j-1)T+i}, \quad i=1,\ldots,T,
\end{align}
as $\hat\Pi_{(j)}$, where $\hat{\overline{\mu}}$ and $\hat{\overline{\Sigma}}$  are defined in \eqref{eq:wa3}. 
The Monte Carlo based approximation of $\overline \Pi$ using $kT$ MCMC draws resulting from Algorithm \ref{losp-sampler} is denoted as $\hat{ \overline \Pi}$,  and equals the uniform mixture of $\hat\Pi_{(j)}$, $j=1,\ldots,k$. Similarly, we define $\tilde{\overline{\Pi}}$ to be the uniform discrete distribution on the observations 
\[
{\overline{\mu}} + {\overline{\Sigma}}^{\frac{1}{2}}\theta'_{i}, \quad i=1,\ldots,kT,
\]
where ${\overline{\mu}} = \frac{1}{k}\sum\limits_{j=1}^k {\mu}_{(j)}$ and ${\overline{\Sigma}}$ is defined  in \eqref{eq:wa2}. Let $Q$ be the probability measure corresponding to a $n$ sample from $P_{\theta_0}$ and $\theta_i'$ ($i=1,\ldots,kT$). Then, the following theorem defines the Monte Carlo error as $W_2(\tilde{\overline \Pi}, \hat{\overline \Pi})$ and quantifies its rate of decay as $n,T$ tend to infinity.
\begin{theorem}\label{mc-err}
Let $\hat{\overline \Pi}$ and let $\tilde{\overline \Pi}$ be as defined above. Under Assumptions \ref{ab}--\ref{mc1}, $n=o(\sqrt{T})$ and $n \rightarrow \infty$, 
\begin{align*}
  W_2^2( \tilde{\overline \Pi}, \hat{\overline \Pi}) = o_{Q_{}}(n^{-1}).
\end{align*}
\end{theorem}

\section{Experiments}

\subsection{Setup}
\label{sim-setup}

We evaluate the performance of Algorithm \ref{losp-sampler} using DPMC as the benchmark for distributed Bayesian inference. The efficiency and accuracy of Algorithm \ref{losp-sampler} and DPMC are evaluated relative to the full data posterior distribution. We use Stan for obtaining parameter draws from the full data and subset posterior distributions \citep{Caretal17}. Our simulated and real data analyses focus on three GLMs: logistic, negative binomial, and multinomial-logistic regressions. All three models are supported by default in Stan. The draws from the subset posterior density in \eqref{eq:subsamp} is obtained using the user-defined probability functions feature of Stan \citep[Section 18.5]{Sta20}. Every sampling algorithm in our simulated and real data analyses runs for 10,000, and we discard the first 5000 draws as burn-ins and thin the chain by collecting every fifth draw. The convergence of the chains is assessed using trace plots. We collect all subset posterior draws of parameters and combine them using Algorithm \ref{losp-sampler} and DPMC. 

We use two metrics for comparing the performance of  Algorithm \ref{losp-sampler} and DPMC. Let $\hat \Pi$ be the MCMC-based approximation of the full data posterior using Stan and $\check \Pi$ be the approximation of $\hat \Pi$ obtained using Algorithm \ref{losp-sampler} or DPMC. Motivated from \eqref{eq:t1e1}, the first metric quantifies the approximation error in using $\check \Pi$ instead of $\hat \Pi$ for inference on parameters as
\begin{align}
  \label{eq:sim1}
  \text{Approximation Error} =  \left[ \| \hat \mu - \check \mu \|^2_2 + \tr \left\{ \hat \Sigma + \check \Sigma - 2 \left( \hat \Sigma^{1/2} \check \Sigma \hat \Sigma^{1/2} \right)\right\}  \right]^{1/2},  
\end{align}
where $\hat \mu$, $\check \mu$ and $\hat \Sigma$, $\check \Sigma$ are the means and covariance matrices of $\hat \Pi$ and $\check \Pi$. The approximation error in \eqref{eq:sim1} is small when the differences between the posterior means and covariance matrices of $\hat \Pi$ and $ \check \Pi$ are small.  The second metric measures the computational gain in using Algorithm \ref{losp-sampler} or DPMC. Let $\hat t$, $\check t$ be the wall-clock run times for obtaining draws from $\hat \Pi$, $\check \Pi$, then the computational gain is defined as $\hat t / \check t$. 

\subsection{Simulated Data Analysis}
\label{sim-set1}

Our simulated data analyses includes logistic and negative binomial regressions. Following the notation in Section \ref{log-reg}, we set $s_i$ to be 15 for every sample, simulate the entries of $X$ as independently from $N(0, 1)$, and set the entries of $\beta$ to alternate between $-2$ and $2$. We vary $n$ as $10^4, 10^5$, vary $p$ as $10, 20$, and simulate $y_i$ as $\text{Binom}\{s_i, 1 / (1 + e^{-x_i^\top \beta})\}$ for every combination of $n$ and $p$. For negative binomial regression, the setup for simulating $X$ and the values of $n,p$ are the same but the $\beta$ entries are set to $-1$ and 1 alternately. The observation $y_i$ is simulated from a negative binomial distribution with  mean parameter $e^{x_i^T \beta}$ and $\phi = 2$ as the overdispersion parameter. This simulation setup is replicated 10 times. The posterior draws of $\beta$ in logistic regression and $\beta, \phi$ in negative binomial regression conditioned on the full data are obtained in every replication using Stan. 

We obtain subset posterior draws for $k=20$ and 50 for DPMC and the proposed method. First, we randomly partition the samples into $k$ disjoint subsets for every $n$ and $p$.  The values of $k$ are small relative to $m$ when $n=10^4$ or $n=10^5$, which satisfies conditions in Theorem \ref{stat-ord}. The subset posterior samplers use Stan to draw $\beta$ in logistic regression and $\beta, \phi$ in negative binomial regression on all the $k$ subsets in parallel from the modified density in \eqref{eq:subsamp}. The post burn-in $\beta$ and $\phi$ draws are collected from the $k$ subsets and aggregated using Algorithm \ref{losp-sampler} and DPMC. 

The approximation errors of the proposed method and DPMC are very similar across all simulation settings (Tables \ref{tbl1}). For both choices of $p$ and $k$, the accuracy of DPMC and proposed method increases with $n$. The accuracy is insensitive to the choice $k$ when $n=10^5$ and decreases slightly moving $k=20$ to $k=50$ when $n=10^4$. The decrease in accuracy happens because the subset sample sizes are much smaller relative to $p$ when $n = 10^4$ and it results in relatively less accurate estimation of subset posterior densities compared to the case when $n=10^5$. Furthermore, the subset sample size decreases as $k$ changes from 20 to 50, but the decrease is more severe when $n=10^4$, resulting in slightly lower accuracy for the $k=50$ case; however, across all simulation settings, accuracy of the proposed method and DPMC agree closely. 

The computational gains of the proposed method and DPMC are very similar across all simulation settings (Tables \ref{tbl2}). When $n=10^4$ and $k=20$ in negative binomial regression, the subset sample size is only slightly smaller than $n$ and the time required for subset and full data posterior computations are very similar. Due to the  extra time spent in combining the samples, there is no gain in efficiency using DPMC or  the proposed method. Except this setting, for every choice of $n$, $p$, and $k$, the computation gains for DPMC and the proposed method are significant. Unlike the statistical accuracy, the computational gains are insensitive to the choice of $m$, $p$, or $k$. We observe that the computational gains increase with $n$, showing the practical advantages of DPMC and the proposed method in massive data settings. We conclude from this simulation study that the proposed method offers similar accuracy and computational gains as DPMC.

\begin{table}[ht]
  \centering
  \begin{tabular}{|r|cc|cc|cc|cc|}
    \hline
    & \multicolumn{8}{c|} {Logistic Regression} \\
    \hline
    & \multicolumn{4}{c|}{$k=20$} & \multicolumn{4}{c|}{$k=50$} \\
    \hline
    & \multicolumn{2}{c|}{$p=10$} & \multicolumn{2}{c|}{$p=20$}& \multicolumn{2}{c|}{$p=10$} & \multicolumn{2}{c|}{$p=20$}\\
    \hline
    & Proposed & DPMC & Proposed & DPMC & Proposed & DPMC & Proposed & DPMC \\ 
    \hline
    $n = 10^4$ & 0.0535 & 0.0535 & 0.3045 & 0.3049 & 0.1206 & 0.1206 & 0.5022 & 0.5027 \\ 
    \hline
    $n = 10^5$ & 0.0457 & 0.0457 & 0.2540 & 0.2543 & 0.0885 & 0.0885 & 0.3678 & 0.3681 \\ 
    \hline
    & \multicolumn{8}{c|}{Negative Binomial Regression} \\
    \hline
    & \multicolumn{4}{c|}{$k=20$} & \multicolumn{4}{c|}{$k=50$} \\
    \hline
    & \multicolumn{2}{c|}{$p=10$} & \multicolumn{2}{c|}{$p=20$}& \multicolumn{2}{c|}{$p=10$} & \multicolumn{2}{c|}{$p=20$}\\
    \hline
    & Proposed & DPMC & Proposed & DPMC & Proposed & DPMC & Proposed & DPMC \\ 
    \hline
    $n = 10^4$ & 0.1440 & 0.1440 & 0.2029 & 0.2029 & 0.3646 & 0.3647 & 0.6389 & 0.6390 \\ 
    \hline
    $n = 10^5$ & 0.0153 & 0.0153 & 0.0225 & 0.0226 & 0.0273 & 0.0273 & 0.0415 & 0.0415 \\ 
    \hline
  \end{tabular}
\caption{Approximation errors in logistic and negative binomial regressions for the different choices of $n, k$, and $p$.}
\label{tbl1}
\end{table}

\begin{table}[ht]
  \centering
  \begin{tabular}{|r|cc|cc|cc|cc|}
    \hline
    & \multicolumn{8}{c|} {Logistic Regression} \\
    \hline
    & \multicolumn{4}{c|}{$k=20$} & \multicolumn{4}{c|}{$k=50$} \\
    \hline
    & \multicolumn{2}{c|}{$p=10$} & \multicolumn{2}{c|}{$p=20$}& \multicolumn{2}{c|}{$p=10$} & \multicolumn{2}{c|}{$p=20$}\\
    \hline
    & Proposed & DPMC & Proposed & DPMC & Proposed & DPMC & Proposed & DPMC \\ 
    \hline
    $n = 10^4$ &  5.2924 &  5.2946 &  6.9790 &  6.9854 &  6.2118 &  6.2190 &  8.1364 &  8.1554 \\ 
    \hline
    $n = 10^5$ & 17.0183 & 17.0211 & 28.8525 & 28.8578 & 37.2964 & 37.3224 & 72.3282 & 72.3988 \\ 
    \hline
    \hline
    & \multicolumn{8}{c|}{Negative Binomial Regression} \\
    \hline
    & \multicolumn{4}{c|}{$k=20$} & \multicolumn{4}{c|}{$k=50$} \\
    \hline
    & \multicolumn{2}{c|}{$p=10$} & \multicolumn{2}{c|}{$p=20$}& \multicolumn{2}{c|}{$p=10$} & \multicolumn{2}{c|}{$p=20$}\\
    \hline
    & Proposed & DPMC & Proposed & DPMC & Proposed & DPMC & Proposed & DPMC \\ 
    \hline
    $n = 10^4$ & 0.9036 &  0.9088 &  1.2580 &  1.2671 &  1.6447 &  1.6512 &  2.2530 &  2.2685 \\ 
    \hline
    $n = 10^5$ & 8.0848 &  8.1086 & 13.0977 & 13.1485 & 14.5329 & 14.7139 & 26.5668 & 26.9506 \\ 
    \hline    
  \end{tabular}
\caption{Computational gains in logistic and negative binomial regressions for the different choices of $n, k$, and $p$.}
\label{tbl2}
\end{table}

\subsection{MovieLens Data Analysis}
\label{sim-set2}

We use MovieLens ratings data with 1 million ratings to illustrate the application of Algorithm \ref{losp-sampler}. This data contains the ratings for about 65 thousand movies from about 72 thousand users of the online movie ratings database named MovieLens. Every record in the database contains information about the user, movie name, rating of the movie by the user ranging from 0.5 to 5 in the increments of 0.5, time of the rating, and genres of the movie rated, where a movie can belong to one or more of the 17 predefined genres. The response is defined to be integers from 1 to 5, where the fractional ratings are rounded up. 

We have also added predictors capturing popularity of movie, the mood of the user, and mapped the 17 genres to 4 movie categories following \citet{Per17}. The four movie categories represent action, children, comedy, and drama genres, and they are represented using dummy variables, with action category being the baseline. The popularity of a movie is defined as logit $\{(n_{\text{lik}} + 0.5) / (n_{\text{rat}} + 1.0)\}$, where logit($x$)$= \log \{x / (1-x)\}$ and $n_{\text{lik}}$ and $n_{\text{rat}}$ are the number of users who rated the movie above 3 and who rated the movie in 30 or fewer most recent observations for the movie. The user's mood equals 1 if the previous movie rating assigned by the user is above 3. Finally, we evaluate the performance of the proposed method and DPMC over ten replications, where each replication has $10^5$ sample size and the full data are divided randomly into 50 subsets. 

We fit multinomial logistic regression to this data due to the five levels of the response. We set the observed rating 5 as the baseline and all the regression coefficients for this response as 0. If the observed rating is $j$ ($j=1, \ldots, 4$), then denote the regression coefficients for the intercept, the three movie category dummy variables with the action category as the baseline, movie popularity, and user's mood as $\beta_j = (\beta_{j1}, \ldots, \beta_{j6})^\top$.  The multinomial logistic regression assumes that 
\begin{align}
  \label{eq:multi-logit}
  \log \frac{\text{Pr} (y_i = j)}{\text{Pr} (y_i = 5)} = x_i^\top \beta_j, \quad j = 1, \ldots, 4, \quad i = 1, \ldots, n. 
\end{align}
This model is supported by Stan, and we use it to obtain posterior draws of $\beta_1, \beta_2, \beta_3, \beta_4$ for the full data. To apply DPMC and the proposed method, the full data are randomly partitioned into 50 subsets and Stan is used to draw parameters from the subset posterior density defined in \eqref{eq:subsamp}. 
The subset posterior draws are combined using DPMC and Algorithm \ref{losp-sampler}. 

Agreeing with the simulation results, DPMC and the proposed method have similar approximation errors and computational gains in inference on $\beta_1, \beta_2, \beta_3, \beta_4$ (Table \ref{tab3}). The approximation errors are small for both methods and matches our theoretical result in \ref{stat-ord}. Similarly, the compuational gains are also $O(k)$ for both methods as predicted by our theory. We conclude that the proposed method is a promising alternative to DPMC for \dnc inference in massive data settings and it offers comparable accuracy and computational gains.  

\begin{table}[ht]
  \centering
{\begin{tabular}{|r|cc|cc|cc|}
    \hline
    & \multicolumn{2}{c|}{Approximation Error} & \multicolumn{2}{c|}{Computational Gain} \\ 
    \hline
    & Proposed & DPMC & Proposed & DPMC \\ 
    \hline
    $\beta_1$ &  0.0678 &  0.0679 & 24.9129 & 24.9560 \\ 
    $\beta_2$ &  0.0341 &  0.0341 & 24.9057 & 24.9560 \\ 
    $\beta_3$ &  0.0158 &  0.0158 & 24.9070 & 24.9559 \\ 
    $\beta_4$ &  0.0157 &  0.0157 & 24.9232 & 24.9561 \\ 
    \hline
  \end{tabular}
}%
\caption{Approximation errors and computational gains in MovieLens data analysis.}
\label{tab3}
\end{table}

\section{Discussion}

We have presented an algorithm for computing an approximation to the WASP based on a location-scatter family. Our simulations in Section \ref{sim-set1} show that if $m$ is large relative to $k$ and $p$, then the approximate algorithm can be used for scalable and accurate binomial and negative binomial regressions in massive data settings. We expect that these advantages carry over to models with random effects; therefore, it is interesting to investigate the analogues of Theorems \ref{stat-ord} and \ref{mc-err} in such models.

\section*{Acknowledgement}
Sanvesh Srivastava's research is partially supported by grants from the Office of Naval Research (ONR-BAA N000141812741) and the National Science Foundation (DMS-1854667/1854662). The R and Stan codes used in our simulated and real data analyses are available at \texttt{https://github.com/blayes/location-scatter-wasp}.

\clearpage

\renewcommand\thesection{\arabic{section}}
\renewcommand\thesubsection{\thesection.\arabic{subsection}}
\renewcommand\thesubsubsection{\thesubsection.\arabic{subsubsection}}

\setcounter{section}{0}

\begin{center}
\textbf{\LARGE Supplementary Material for An Algorithm for Distributed Bayesian Inference in Generalized Linear Models}
\end{center}

\section{Proof of Theorem 1}
\label{sec:appendix2}

Assumption 2 implies that $\overline \Pi(\cdot \mid \Dcal)$ and  $\Pi(\cdot \mid \Dcal)$ belong to the same location-scatter family, and using Theorem 2.3 in \cite{Alvetal16} we have that 
\begin{align}
  \label{eq:s1e1}
  W^2_2\{\Pi(\cdot \mid \Dcal), \overline \Pi(\cdot \mid \Dcal) \} = \|\mu - \overline \mu  \|_2^2 +
   \tr\{\Sigma + \overline \Sigma - 2 (\overline \Sigma^{1/2} \Sigma \overline \Sigma^{1/2})^{1/2}\}.
\end{align}
If $k = O(1)$, then Lemma \ref{mean-diff} and Lemma \ref{var-diff} together imply that the expression on the right is $o(n^{-1})$ in $P_{\theta_0}^n$-probability. 
The theorem is proved.

\section{Proof of Theorem 2}

We begin by observing that by the definition of the Wasserstein distance we have,
\begin{align}\label{msplit}
W_2^2( \tilde{\overline \Pi}, \hat{\overline \Pi})&\leq \frac{1}{k} \sum\limits_{j=1}^k \frac{1}{T} \sum\limits_{i=1}^T \Vert a_j + B_j \theta'_{(j-1)T+i}\Vert^2\\
&\leq \frac{2}{k}\sum_{j=1}^k \Vert a_j \Vert^2 +\frac{2}{kT}\sum_{j=1}^k\sum_{i=1}^T \Vert B_j \theta'_{(j-1)T+i}\Vert ^2\\
&\leq \frac{2}{k}\sum_{j=1}^k \Vert a_j \Vert^2 + 2 \max_{1\leq j \leq k}\Vert B_j \Vert^2 \frac{1}{kT}\sum_{i=1}^{kT} \Vert \theta'_{i}\Vert ^2
\end{align}
where $\|B\|$ for a matrix $B$ denotes the operator norm, 
\[
a_j:=\hat{\overline{\mu}} -{\overline{\mu}} + \hat{\overline{\Sigma}}^{\frac{1}{2}} \hat\Sigma_{(j)}^{-{1}/{2}}(\mu_{(j)}-\hat{\mu}_{(j)}), \hbox{ and } B_j:=\hat{\overline{\Sigma}}^{\frac{1}{2}} \hat\Sigma_{(j)}^{-1/2}\Sigma_{(j)}^{\frac{1}{2}}-{\overline{\Sigma}}^{1/2}.
\]
Note that by the law of large numbers, in view of the above, it suffices to show that 
\begin{equation}\label{shy:eqn1}
a_j=o_{Q_{}}(n^{-1/2}), \quad \hbox{and }  \quad \Vert B_j\Vert =o_{Q_{}}(n^{-1/2}).    
\end{equation}

We begin by establishing the former statement, and towards this note that since $\Vert \hat{{\mu}}_{(j)} -{{\mu}}_{(j)} \Vert=O_{Q_{}}(T^{-1/2})$, we have 
\[
\Vert \hat{\overline{\mu}} -{\overline{\mu}} \Vert \leq \frac{1}{k}\sum\limits_{j=1}^k \Vert \hat{{\mu}}_{(j)} -{{\mu}}_{(j)} \Vert=O_{Q_{}}(T^{-1/2}).  
\]
Moreover, we note  that
\begin{align*}
    \Vert \hat{\overline{\Sigma}}^{\frac{1}{2}} \hat\Sigma_{(j)}^{-1/2} \Vert &\leq 
    \Vert n\hat{\overline{\Sigma}} \Vert^{\frac{1}{2}} \Vert (n\hat\Sigma_{(j)})^{-1} \Vert^{1/2}\\
    &\leq \left(\sqrt{\frac{1}{k} \sum\limits_{j=1}^k \Vert n\hat{{\Sigma}}_{(j)}\Vert } \right)\Vert (n\hat\Sigma_{(j)})^{-1} \Vert^{1/2} \quad \hbox{(Theorem 9 of \citet{Bhatia18})}\\
    &= O_{Q_{}}(1).
\end{align*}
%\begin{align*}
%    \Vert \hat{\overline{\Sigma}}^{\frac{1}{2}} \hat\Sigma_{j}^{-1/2} \Vert &\leq 
%    \Vert n\hat{\overline{\Sigma}} \Vert^{\frac{1}{2}} \Vert n\hat\Sigma_{j} \Vert^{-1/2}\\
%    &\leq \left(\sqrt{\frac{1}{k} \sum\limits_{j=1}^k \Vert n\hat{{\Sigma}}_{j}\Vert } \right)\Vert n\hat\Sigma_{j} \Vert^{-1/2} \quad \hbox{(Theorem 9 of \citet{Bhatia18})}\\
%    &= O_{Q_{}}(1).
%\end{align*}
Combining the previous two observations, we have
\begin{align*}
\Vert a_j \Vert^2 &\leq  2\Vert \hat{\overline{\mu}} -{\overline{\mu}} \Vert^2 + 2\Vert \hat{\overline{\Sigma}}^{\frac{1}{2}} \hat\Sigma_{(j)}^{-{1}/{2}}(\mu_{(j)} -\hat{\mu}_{(j)})\Vert^2\\
&\leq O_{Q_{}}(T^{-1}) + 2 \Vert \hat{\overline{\Sigma}}^{\frac{1}{2}} \hat\Sigma_{(j)}^{-{1}/{2}} \Vert^2 \Vert \mu_{(j)} -\hat{\mu}_{(j)} \Vert^2 \\
&=O_{Q_{}}(T^{-1}) + O_{Q_{}}(1)O_{Q_{}}(T^{-1})= O_{Q_{}}(T^{-1})=o_{Q_{}}(n^{-1/2}).
\end{align*}

Now we establish the second statement of \eqref{shy:eqn1}. Towards this end we note that, 
\begin{align}
\Vert B_j\Vert 
%&\leq n^{-1/2} \left( \Vert(n\hat{\overline{\Sigma})}^{\frac{1}{2}}_T \hat\Sigma_{j,T}^{-1/2}\hat\Sigma_j^{\frac{1}{2}}-I_{\theta_0}^{-1/2}\Vert +  \Vert I_{\theta_0}^{-1/2} - (n\hat{\overline{\Sigma}})^{1/2}\Vert\right) \\
&\leq n^{-1/2} 
\left( 
\Vert I_{\theta_0}^{-1/2} - (n{\overline{\Sigma}})^{1/2}\Vert  
+ \Vert I_{\theta_0}^{-1/2}\hat\Sigma_{(j)}^{-1/2}\Sigma_{(j)}^{\frac{1}{2}} - I_{\theta_0}^{-1/2} \Vert 
+\Vert(n\hat{\overline{\Sigma}})^{\frac{1}{2}} \hat\Sigma_{(j)}^{-1/2}\Sigma_{(j)}^{\frac{1}{2}}-I_{\theta_0}^{-1/2}\hat\Sigma_{(j)}^{-1/2}
\Sigma_{(j)}^{\frac{1}{2}}\Vert \right)
\nonumber \\
&\leq n^{-1/2} \left( 
\Vert I_{\theta_0}^{-1/2} - (n{\overline{\Sigma}})^{1/2}\Vert 
+ \Vert I_{\theta_0}^{-1/2}\Vert \Vert\hat\Sigma_{(j)}^{-1/2}\Sigma_{(j)}^{\frac{1}{2}} - I \Vert
+\Vert(n\hat{\overline{\Sigma}})^{\frac{1}{2}} -I_{\theta_0}^{-1/2}\Vert\Vert\hat\Sigma_{(j)}^{-1/2}\Sigma_{(j)}^{\frac{1}{2}}\Vert
\right), \label{shy:eqn2} 
\end{align}
where we have used $I_{\theta_0}$ to denote the Fisher information matrix. Note that it suffices to show that the term within parenthesis in \eqref{shy:eqn2} is $o_{_{Q_{}}}(1)$. For the first term we observe using Lemma \ref{slemma}, Lemma \ref{shy:lem2}, and  \eqref{shy:eqn3} that, 
\begin{align}
\Vert (n{\overline{\Sigma}})^{\frac{1}{2}}-I_{\theta_0}^{-1/2}\Vert  \leq  \sqrt{\Vert(n{\overline{\Sigma}})-I_{\theta_0}^{-1}\Vert} & \leq \sqrt{\Vert(n{\overline{\Sigma}})-I_{\theta_0}^{-1}\Vert_F}\nonumber \\
&\leq \sqrt{d((n{\overline{\Sigma}}),I_{\theta_0}^{-1}) \left[ \sqrt{\tr(n{\overline{\Sigma}}) } + \sqrt{\tr(I_{\theta_0}^{-1})}\right]}\nonumber \\
&\leq \sqrt{d((n{\overline{\Sigma}}),I_{\theta_0}^{-1})} \sqrt{
\left[ \sqrt{\frac{1}{k}\sum_{j=1}^k\tr(n{\Sigma}_{(j)}) } + \sqrt{\tr(I_{\theta_0}^{-1})}\right]}\nonumber \\
&= o_{{Q_{}}}(1) \times O_{_{Q_{}}}(1)=o_{_{Q_{}}}(1). \label{shy:eqn4}
\end{align}
For the second term within parenthesis in \eqref{shy:eqn2}, we note that our Assumption 5 implies that
\[
\Vert\hat\Sigma_{(j)}^{-1/2}\Sigma_{(j)}^{\frac{1}{2}} - I \Vert \leq \sqrt{\Vert \hat\Sigma_{(j)}^{(-1)}\Vert} \Vert\Sigma_{(j)}^{\frac{1}{2}} - \hat\Sigma_{(j)}^{1/2} \Vert \leq 
\sqrt{\Vert (n\hat\Sigma_{(j)})^{(-1)}\Vert} \sqrt{\Vert n\Sigma_{(j)} - n\hat\Sigma_{(j)} \Vert_F}=o_Q(1),
\]
%\[
%\Vert\hat\Sigma_{j}^{-1/2}\Sigma_j^{\frac{1}{2}} - I \Vert \leq \sqrt{\Vert \hat\Sigma_{j}\Vert^{(-1)}} \Vert\Sigma_j^{\frac{1}{2}} - \hat\Sigma_{j}^{1/2} \Vert \leq 
%\sqrt{\Vert n\hat\Sigma_{j}\Vert^{(-1)}} \sqrt{\Vert n\Sigma_j - n\hat\Sigma_{j} \Vert_F}=o_Q(1),
%\]
where $\| B\|_F$ is the Frobenius norm of the matrix $B$. 
For the last term within parenthesis in \eqref{shy:eqn2}, an argument mimicking that in \eqref{shy:eqn4} and using Assumption 5 confirms that it is of order $o_Q(1)$. 

\section{Technical Lemmas}
\label{sec:technical-lemmas}

The following lemma states that asymptotic order of the first term on the right hand side in \eqref{eq:s1e1}. 
\begin{lemma} \label{mean-diff}
  Let $\mu$ and $\overline \mu$ be the means of $\Pi(\cdot \mid \Dcal)$ and $\overline \Pi(\cdot \mid \Dcal)$. If Assumptions 1--4 hold and $k = O(1)$, then as $n, m \rightarrow \infty$
  \begin{align*}
    \|\mu - \overline \mu  \|_2^2 = o(n^{-1}) \text{ in }     P_{\theta_0}^n\text{-probability}.
  \end{align*}
\end{lemma}
\begin{proof}
  The proof follows from the proof Theorem 1 in \citet{XueLia17} because our assumptions include all the regularity assumptions required for Theorem 1 in \citet{XueLia17} to hold. 
\end{proof}

In the following, we define $d(\cdot,\cdot)$ as 
\[
d(A,B):=\sqrt{\tr\left(A+B-2(A^{1/2}BA^{1/2})^{1/2}\right)},
\]
where $A$ and $B$ are two $p\times p$ positive semidefinite matrices. In \citet{Bhatia18} (see page 3 therein) it is shown that $d(\cdot,\cdot)$ defines a metric on the space of positive semidefinite matrices.  By the Wasserstein mean of $K$ positive semidefinite matrices $A_{k}$, $k=1,\ldots,K$,  we mean the the variance-covariance matrix of the Wasserstein barycenter of $N(0,A_k)$, $k=1,\ldots,K$.    

\begin{lemma}\label{slemma}
Let $A_{k}$, $k=1,\ldots,K$ be a sequence of $p\times p$ positive definite matrices, and let their Wasserstein mean be denoted by $\bar{A}$. Then for another positive definite matrix $A_0$ we have, 
\begin{equation}\label{sineq}
d(\bar{A},A_0) \leq 2 \sqrt{\frac{p}{K} \sum_{k=1}^K \Vert A_k-A_0\Vert}  \leq 2 \sqrt{\frac{p}{K} \sum_{k=1}^K \Vert A_k-A_0\Vert_{F}}.  
\end{equation}

\end{lemma}
\begin{proof}
By the definition of $\bar{A}$, or see (57) in \citet{Bhatia18}, we have
\[
\bar{A}:=\argmin_{X\succ 0} \sum_{k=1}^K d^2(X,A_k).
\]
This implies that 
\[
\frac{1}{K}\sum_{k=1}^K d^2(\bar{A},A_k)\leq\frac{1}{K}\sum_{k=1}^K d^2(A_0,A_k).     
\]
Now we have by use of the triangle inequality and the AM-GM inequality that 
\begin{align}
d^2(\bar{A},A_0)&\leq 2\left[ \frac{1}{K}\sum_{k=1}^K d^2(\bar{A},A_k) + \frac{1}{K}\sum_{k=1}^K d^2(A_0,A_k)\right]\\
&\leq \frac{4}{K} \sum_{k=1}^K d^2(A_0,A_k)\leq \frac{4}{K} \sum_{k=1}^K \Vert A_0^{1/2} -A_k^{1/2}\Vert_{F}^2,\\
&\leq \frac{4p}{K} \sum_{k=1}^K \Vert A_0^{1/2} -A_k^{1/2}\Vert^2,
\end{align}
where the last inequality follows from Theorem 1 of \citet{Bhatia18}. Now since $\sqrt{\cdot}$ is operator monotone, using Theorem X.1.1 of \citet{Bhatia13} with the above inequality yields the first inequality of \eqref{sineq}. The final inequality of \eqref{sineq} follows by the fact that the Frobenius norm upper bounds the operator norm.  
\end{proof}

\begin{lemma}\label{shy:lem2}
 For two $p\times p$ positive semi-definite matrices $A$ and $B$, we have
\[
\Vert A -B\Vert_{F}\leq d(A,B) \left( \sqrt{\tr[A]}+\sqrt{\tr[B]}\right)
\]
\end{lemma}
\begin{proof}
Let us define for any positive semi-definite $p\times p$ matrix $C$, 
\begin{equation*}
{\cal F}(C):=\{M_{p\times p}: C=MM^\T \}.
\end{equation*}
Let M,N be members of ${\cal F}(A)$ and ${\cal F}(B)$, respectively. Then we have,
\begin{align*}
\Vert A - B\Vert_F &= \Vert MM^\T - NN^\T \Vert_F    \\
&= \Vert MM^\T - MN^\T + MN^\T - - NN^\T \Vert_F    \\
&\leq \Vert M\Vert_F \Vert M^\T-N^\T\Vert_F + \Vert N^\T \Vert_F  \Vert M-N\Vert_F\\
&= \Vert M-N\Vert_F (\Vert M\Vert_F + \Vert N\Vert_F)\\
&= \Vert M-N\Vert_F \left(\sqrt{\tr(A)}+\sqrt{\tr(B)}\right)
\end{align*}
Using the above with Theorem 1 of \citet{Bhatia18} yields,
\begin{align*}
\Vert A - B\Vert_F &\leq \left(\sqrt{\tr(A)}+\sqrt{\tr(B)}\right) \min_{M\in{\cal F}(A);N\in{\cal F}(B)}\Vert M-N\Vert_F \\
&=d(A,B)\left(\sqrt{\tr(A)}+\sqrt{\tr(B)}\right).
\end{align*}

\end{proof}

The following lemma states that asymptotic order of the second term on the right hand side in \eqref{eq:s1e1}. 
\begin{lemma} \label{var-diff}
  Let $\Sigma$ and $\overline \Sigma$ be the covariance matrices of $\Pi(\cdot \mid \Dcal)$ and $\overline \Pi(\cdot \mid \Dcal)$. If Assumptions 1--4 hold, then as $n, m \rightarrow \infty$
  \begin{align*}
    d^2(\Sigma,\overline \Sigma) = o_{P^n_{\theta_0}}(n^{-1}).
  \end{align*}
\end{lemma}
\begin{proof}
  Let $\Sigma_{(j)} = \var(\theta_{(j)} \mid \Dcal_{(j)})$. Assumption 3, the existence of moment generating function, and Theorem 4 in \citet{Kasetal90} imply via the Laplace approximation of the posterior and the subset posteriors that
  \begin{align}
    \label{eq:l2e1}
    \Sigma_{(j)} = \frac{I_{jm}^{-1}}{n} + O_{P^m_{\theta_0}}(n^{-2}), \quad \hbox{and}\quad 
    \Sigma = \frac{I_n^{-1}}{n} + O_{P^n_{\theta_0}}(n^{-2}), 
  \end{align}
  where $I_{jm}$ and $I_{n}$ are the Fisher information matrices evaluated at the  maximum likelihood estimators computed using subset $j$ and full data, respectively.
  Since the maximum likelihood estimates are consistent estimates of $\theta_0$ and matrix inversion is a continuous operator on the subspace of invertible matrices, we have 
  \begin{align*}
    %\label{eq:l2e1a}
  I_{jm}^{-1} = I_{\theta_0}^{-1} + o_{P^m_{\theta_0}}(1), \quad \hbox{and}\quad I_n^{-1} = I_{\theta_0}^{-1} + o_{P^n_{\theta_0}}(1).
  \end{align*}
  Combining the above observations we have 
  \begin{align}\label{shy:eqn3}
        n\Sigma_{(j)}-I_{\theta_0}^{-1} = o_{P^m_{\theta_0}}(1), \quad n\Sigma -I_{\theta_0}^{-1}=o_{P^n_{\theta_0}}(1), \quad  \hbox{and} \quad n\Sigma_{(j)} -  n\Sigma = o_{P^m_{\theta_0}}(1). 
\end{align}
Now using Lemma \ref{slemma} we have 
\begin{align*}
nd^2(n\Sigma,n\overline \Sigma)&=d^2(n\Sigma,n\overline \Sigma)\\
&\leq \frac{4p}{K} \sum_{k=1}^K \Vert n\Sigma_{(j)}-n\Sigma\Vert_{F}= o_{P^n_{\theta_0}}(n^{-1}).  
\end{align*}

\end{proof}

\bibliographystyle{Chicago}
\bibliography{papers}

\end{document}